\documentclass[10pt,conference,a4paper]{IEEEtran}

\usepackage{graphicx}
\usepackage[utf8x]{inputenc}
\usepackage[english]{babel}
\usepackage{amsmath}
\usepackage{amssymb}
\usepackage{color}
\usepackage{cite} 
\usepackage{amsthm}
\usepackage{comment}

\newtheorem{ex}{Example}

\begin{document}

\title{On the Application of Multiuser Detection in Multibeam Satellite Systems}

\author{Giulio Colavolpe, Andrea Modenini,  Amina Piemontese, and Alessandro Ugolini\\
\normalsize  Universit\`a di Parma, Dipartimento di Ingegneria dell'Informazione, Parco Area delle Scienze 181/A, Parma - ITALY\\
}
\maketitle

\begin{abstract}
We study the achievable rates by a single user in multibeam satellite scenarios. We show alternatives to the conventional symbol-by-symbol detection applied at user terminals. Single user detection is known to suffer from strong degradation when the terminal is located near the edge of the coverage area, and when aggressive frequency reuse is adopted.
For this reason, we consider multiuser detection, and take into account the strongest interfering signal. 
Moreover, we analyze a different transmission strategy, where the signals from two adjacent beams jointly serve two users in a time division multiplexing fashion. We describe an information-theoretic framework to compare different transmission/detection strategies by computing the information rate of the user in the reference beam.
\end{abstract}

\section{Introduction}\label{s:intro} 
The recent years have witnessed the explosion of satellite services and applications, and the related growing demand for high data rates. 
Next-generation satellite systems need new technologies to improve their spectral efficiency, in order to sustain the information revolution of modern societies.
The grand challenge is to satisfy this demand by living with the scarcity of the frequency spectrum. 
Resource sharing is probably the only option, and can be implemented by adopting a multibeam system architecture which allows to reuse the available bandwidth in many beams. The interference caused by resource sharing is typically considered undesirable, but a way to dramatically improve the spectral efficiency is to exploit this interference, by using interference management techniques at the receiver. 

In this paper, we consider the forward link of a multibeam satellite system, where an aggressive frequency reuse is applied. Under these conditions,
the conventional single user detector (SUD) suffers from a severe performance degradation when the terminal is located near the edge of the coverage
area, due to the high co-channel interference. On the other hand, the application of a decentralized multiuser detector (MUD) at the terminal which is able to cope with the interference can guarantee the required performance.

The literature on multiuser detection is wide, and in the area of satellite communications essentially focuses on the adjacent channel interference mitigation for the return link~\cite{PiGrCo13,BeElKa02,CoFePi11}, and includes centralized techniques to be applied at the gateway. Less effort has been devoted to the forward link. 

Recently, we investigated in~\cite{AnAnCaCo14} the benefits that can be achieved, in terms of spectral efficiency, when high frequency reuse is applied in a DVB-S2~\cite{Et03} system, and multiuser detection is adopted at the terminal to manage the presence of strong co-channel interference. The superiority of the MUD has been demonstrated through error rate simulations. In~\cite{CoAnPe14}, the authors study the applicability of a low complexity MUD based on soft interference cancellation. The advantage of the proposed detector is shown in terms of frame error rate. 

In this paper, we generalize the analysis of~\cite{AnAnCaCo14} by supplying an information-theoretic framework which allows us to evaluate the performance in terms of information rate (IR), without the need of lengthy 
error rate simulations and hence strongly simplifying the comparison of various scenarios. 
Furthermore, we consider also a different transmission strategy, where the two signals intended for the two beams cooperate to serve the two users (one in the first beam and the other in the second one) in a time division multiplexing fashion. In other words, instead of serving simultaneously the two users in the adjacent beams, the users are served  consecutively in an exclusive fashion. 

The conclusive picture is complex, since our results show that a transmission/detection strategy which is universally superior to the others does not exist, but the performance depends on several factors, such as the signal-to-noise ratio (SNR), the users' power profile, and the rate of the strongest interferer.
This fact outlines the importance of the proposed analysis framework, which can avoid to resort to computationally intensive simulations.

In the following, Section~\ref{s:sys_model} presents the system model and describes the two considered scenarios and related detection strategies. 
The information-theoretic analysis is treated in Sections~\ref{s:scenario1} and~\ref{s:scenario2}, and gives us the necessary means for the computation 
of the information rate for the reference beam. Section~\ref{s:num_res} presents the results of our study, whereas conclusions are drawn in Section~\ref{s:conclusions}.

\section{System Model}\label{s:sys_model}
We focus on the forward link of a satellite communication system. Figure~\ref{fig:system} depicts a schematic view of the baseband model we are considering. 
Signals $s_i(t)$, $i=1,\dots,K$, are $K$ signals transmitted by a multibeam satellite in the same frequency band. The satellite is thus composed of $K$ 
transmitters (i.e. transponders) and serves $K$ users on the ground. 
The nonlinear effects related to the high power amplifiers which compose the transponders are neglected since a multibeam satellite generally works in a multiple carriers per transponder modality, and hence the operational point of its amplifiers is far from saturation.  
We consider the case where the users experience a high level of co-channel interference, since we assume that they are located close to the edge of the coverage area of a beam and that an aggressive frequency reuse is applied.

The signal received by a generic user can be expressed as
 \begin{equation}
 r(t)=\sum_{i=1}^K\gamma_i s_i(t)+w(t)\,, \label{eq:model}
 \end{equation}
where $\gamma_i$ are proper complex gains, assumed known at the receivers, and $w(t)$ is the thermal noise. Without loss of generality, 
we assume that ``User 1'' is the reference user and that $\gamma_i\ge\gamma_{i+1}$. We will evaluate the ultimate performance limit of the reference user 
when the other users adopt fixed rates. We will consider the following two scenarios, that imply different transmission and detection strategies.
\newline
\textbf{Scenario 1.} Signal $ s_i(t) $ is intended for user $i$, and we are interested in the evaluation of the performance for ``User~1'', whose information 
  is carried by the signal with $i$=1. For this scenario, we evaluate the IR, or equivalently the achievable spectral efficiency, 
when ``User 1'' employs different detectors. In particular, we consider the case when  ``User 1'' employs:
\begin{itemize}
\item A SUD. In this case, all interfering signals $s_i(t)$, $i=2,\dots,K$ are considered as if they were additional thermal noise.   
\item A MUD for the useful signal and one interferer. In this case,  the receiver is designed to detect the useful signal and 
the most powerful interfering signal (that with $i$=2 in our model) whereas all the remaining signals are considered as if they were additional thermal noise. 
Data related to the interfering user are discarded after detection. This case will be called MUD$\times$2 in the following. 
\end{itemize}
Our analysis can be easily extended to the case of a MUD designed for more than two users. On the other hand, given the actual users' power profile, it has been shown 
in~\cite{AnAnCaCo14} that the MUD$\times$2 offers the best tradeoff between complexity and performance.
\begin{figure}
	\begin{center}
		\includegraphics[width=1.0\columnwidth]{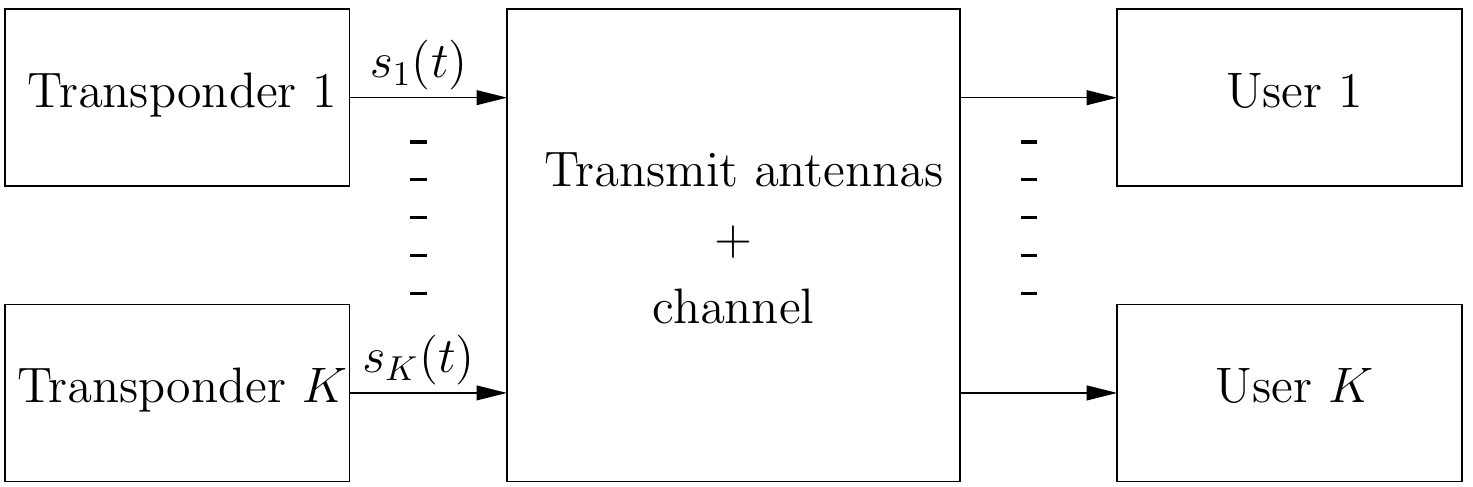}
		\caption{Block diagram of the considered system.}\label{fig:system}
	\end{center}
\end{figure}
\newline
\textbf{Scenario 2.} A different strategy is adopted in this case. Without loss of generality, we will consider detection of signals $s_1(t)$ and $s_2(t)$ and users 1 and 2 only. As in scenario~1, the remaining signals are considered as additional thermal noise.
Instead of simultaneously transmitting  signal $s_1(t)$ to ``User~1'' and signal $s_2(t)$ to ``User~2'', as in the previous scenario, we here serve ``User 1'' first by employing both signals $s_1(t)$ and $s_2(t)$ for a fraction $\alpha$ ($0\le\alpha\le 1$) of the total time, and then ``User~2'' by employing both signals $s_1(t)$ and $s_2(t)$ for the remaining fraction $1-\alpha$ of the total time. The fraction $\alpha$ can be chosen in order to maximize the sum-rate or simply by taking into account the different data rate needs of the users.   

Signals $s_1(t)$ and $s_2(t)$ are independent (although carrying information for the same user) and one of them is properly phase-shifted with respect to the other one in order to maximize the IR\footnote{We are assuming that all transmitted signals are modulated by using the same frequency.}. 
The value of this phase shift can be found by computing the IR for a fixed value of the phase shift and then looking for the value providing the maximum value of the IR. A proper discretization of the phase must be used.
The receiver must jointly detect both signals and its complexity is comparable to that of the MUD$\times$2 described for the first scenario.

% When computing the IR we will consider the following two cases:
% \begin{itemize}
% \item The two signals carry the same information. Particularly, signal $s_1(t)$ is a properly phase-shifted version of $s_2(t)$ in order to allow them to sum 
% up coherently.\footnote{We are assuming that all transmitted signals are modulated by using the same frequency.} In this case, the receiver is designed 
% for only one signal with a proper amplitude. With reference to the previous scenario, its complexity is that of the SUD.
% \item Signals $s_1(t)$ and $s_2(t)$ are independent (although carrying information for the same user) and one of them is properly phase-shifted with respect to 
% the other one in order to maximize the IR. The receiver must jointly detect both signals and its complexity is comparable to that of the MUD$\times$2. 
% \end{itemize}

\section{Information-theoretic Analysis for Scenario~1}\label{s:scenario1}
We first consider multiuser detection and describe how to compute the IR related to ``User 1'' assuming the MUD$\times$2 receiver. 
The same technique can be used to compute the IR related to ``User 2'' and straightforwardly extends to the case of MUD for more than two users.
The channel model assumed by the receiver is
	\begin{equation}
		y = x_1+ \gamma x_2 + w \, ,\label{eq:mac}
	\end{equation}
where $x_i$ is the $M^{(i)}$-ary complex-valued symbol sent over the $i$th beam and $w$ collects the thermal noise, with power $N$, and the remaining interferers 
that the receiver is not able to cope with. Symbols $x_1$ and $x_2$ are mutually independent and distributed according to their probability mass function $P(x_i)$. 
They are also properly normalized such that $\mathrm{E}\{|x_i|^2\}=P$, where $P$ is the transmitted power per user. Parameter $\gamma$ is complex-valued and models 
the power unbalance and the phase shift between the two signals.  Random variable $w$ is assumed complex and Gaussian. We point out that this is an approximation 
exploited only by the receiver, while in the actual channel the interference is clearly generated as in~\eqref{eq:model}. The MUD$\times$2 detector has a 
computational complexity which is proportional to the product $M^{(1)}  M^{(2)}$~\cite{Ve98}.

We are interested here in the computation of the maximum achievable rate $R_1$ for ``User 1'' when ``User~2'' adopts a fixed rate $R_2$, and the MUD$\times$2 is employed. Rates are defined as $R_i=r^{(i)}\log_2(M^{(i)})$, where $r^{(i)}$ is the rate of the adopted binary code.
The rates of the other $K-2$ interferers do not condition our results since at the receiver they are treated just as noise. This problem is quite different with respect to the case of the Multiple Access Channel (MAC) discussed in \cite{CoTh06} where both rates $(R_1,R_2)$ are jointly selected, while here the rate $R_2$ is fixed and data of ``User~2'' can be discarded after detection.

The IR for ``User~1'' in the considered scenario is given by Theorem~\ref{t:ThIR}, whose proof is based on the following two lemmas.
	\newtheorem{Lemma}{Lemma}
	\begin{Lemma}
		For a fixed rate $R_2$, the rate
		\begin{equation}\nonumber
			I_{A}= \begin{cases}
				I(x_1;y|x_2) & \mathrm{if\quad} R_2 < I(x_2;y) \\
				 I(x_1,x_2;y)-R_2 & \mathrm{if\quad} I(x_2;y) \leq R_2<I(x_2;y|x_1) \\
			         0                                          & \mathrm{if\quad} R_2\geq I(x_2;y|x_1)
			       \end{cases}
		\end{equation}
is achievable by ``User 1'' and is not a continuous function of $P/N$. Namely, a cut-off $\mathrm{SNR}_c$  exists such that $I_A=0$ for $P/N \leq \mathrm{SNR}_c$ and
		$I_A>0$ for $P/N > \mathrm{SNR}_c$ with a discontinuity.
	\end{Lemma}
	\begin{proof}
		In \cite{CoTh06}, it is shown that the achievable region for the MAC is given by the region of points ($R_1,R_2$) such that
		\begin{eqnarray}
			R_1 & <& I(x_1;y|x_2)=I_1 \label{eq:c1} \\
			R_2 & <& I(x_2;y|x_1)=I_2 \\
			R_1+R_2 & <& I(x_1,x_2;y)=I_{\rm J}\label{eq:c3}  \,.
		\end{eqnarray}
An example of such a region is shown in Figure~\ref{fig:region}.
\begin{figure}
	\begin{center}
		\includegraphics[width=0.85\columnwidth]{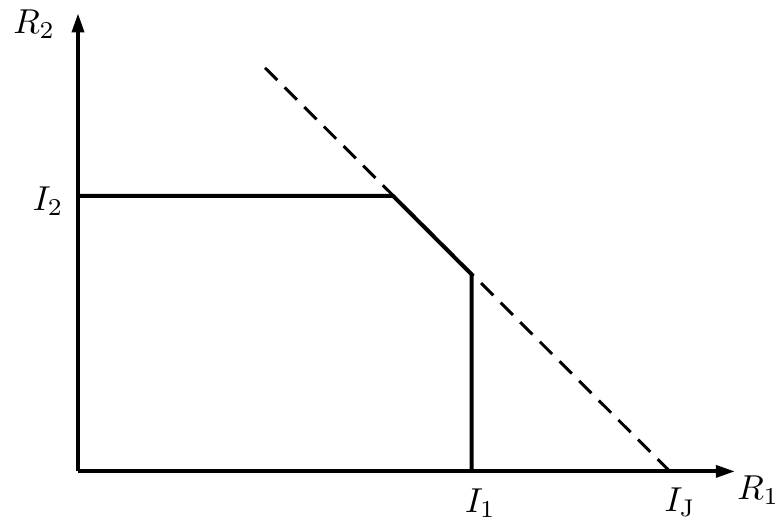}
		\caption{Example of MAC capacity region.}\label{fig:region}
	\end{center}
\end{figure}
If $R_2$ is constrained to a given value, we derive from (\ref{eq:c1}) and (\ref{eq:c3}) that 
$$R_1< \min\{ I(x_1;y|x_2),I(x_1,x_2;y)-R_2\}$$ 
when $R_2<I(x_2;y|x_1)$. The first term is lower when 
$$R_2 < I(x_1,x_2;y) -I(x_1;y|x_2) = I(x_2;y)\,.$$ 
Thus, $I_{A}$ is an achievable rate for ``User 1''.
		
We now prove that $I_A$ has a cut-off rate. Since, $I(x_2;y|x_1)$ is a non-decreasing function of $P/N$~\cite{GuShVe05}, there exists $\mathrm{SNR}_c$ such that $I(x_2;y|x_1)=R_2$, and hence $$I_A(\mathrm{SNR}_c)=0.$$ On the other hand for a small $\varepsilon>0$,  it holds $R_2=I(x_2;y|x_1)-\delta$ where $\delta>0$. It follows that 
$I(x_1;y|x_2)>I(x_1,x_2;y)-R_2$. Thus 
$$I_A(\mathrm{SNR}_c+\varepsilon)=I(x_1,x_2;y)-R_2> I(x_1;y)>0$$ 
for $\varepsilon\rightarrow0^+$.
\qedhere
\end{proof}
\textit{Discussion}: The proof of the lemma can be done graphically by considering the intersection of the achievable
region with a horizontal line at height $R_2$. 

When $R_2>I(x_2;y|x_1)$ clearly the rate of ``User 2'' cannot be achieved. However, we also have to account for this case and therefore we consider also the achievable rate $I(x_1;y)$, which is the relevant rate when ``User 2'' is just considered as interference. In this case, the receiver exploits the statistical knowledge of the signal $s_2(t)$ but does not attempt to recover the relevant information.
\begin{Lemma}
		The rate $I_S(P/N)=I(x_1;y)$ as a function of $P/N$ is always greater than 0 and satisfies
		\begin{eqnarray*}
			I_S(\mathrm{SNR}_c) & = & \lim_{\varepsilon \rightarrow0^+} I_A(\mathrm{SNR}_c+\varepsilon) \\
			I_S(\mathrm{SNR}_c+\delta) & < & I_A(\mathrm{SNR}_c+\delta)
		\end{eqnarray*}
		for any $\delta>0$.

	\end{Lemma}
	\begin{proof}
		The proof is straightforward. It can be done by observing that $I(x_1;y)\leq I(x_1;y|x_2)$
		and that $I(x_1;y)\leq I(x_1,x_2;y)$.
	\end{proof}

	\newtheorem{Theorem}{Theorem}
	\begin{Theorem}\label{t:ThIR}
		The achievable information rate for a single user on the two users multiple access channel, for a fixed rate $R_2$,
		is given by
		\begin{equation}\label{eq:thm1}
			R_1 < \max\{I_S,I_A\}\, ,
		\end{equation}				
		and is a continuous function of $P/N$.
	\end{Theorem}
	\begin{proof}
	 Proof made by means of the Lemmas. In fact, $I_A$ and $I_S$ are the maximum rates achievable by ``User 1'' when ``User 2'' can be perfectly decoded, or not. 
	 An alternative graphical proof can be derived from Figure~\ref{fig:region_thm}, which plots the rate achievable by ``User~1'' as a function of $R_2$, for a generic fixed value of $P/N$. We clearly see that inequality~\eqref{eq:thm1} holds.
	\end{proof}
\begin{figure}
	\begin{center}
		\includegraphics[width=0.95\columnwidth]{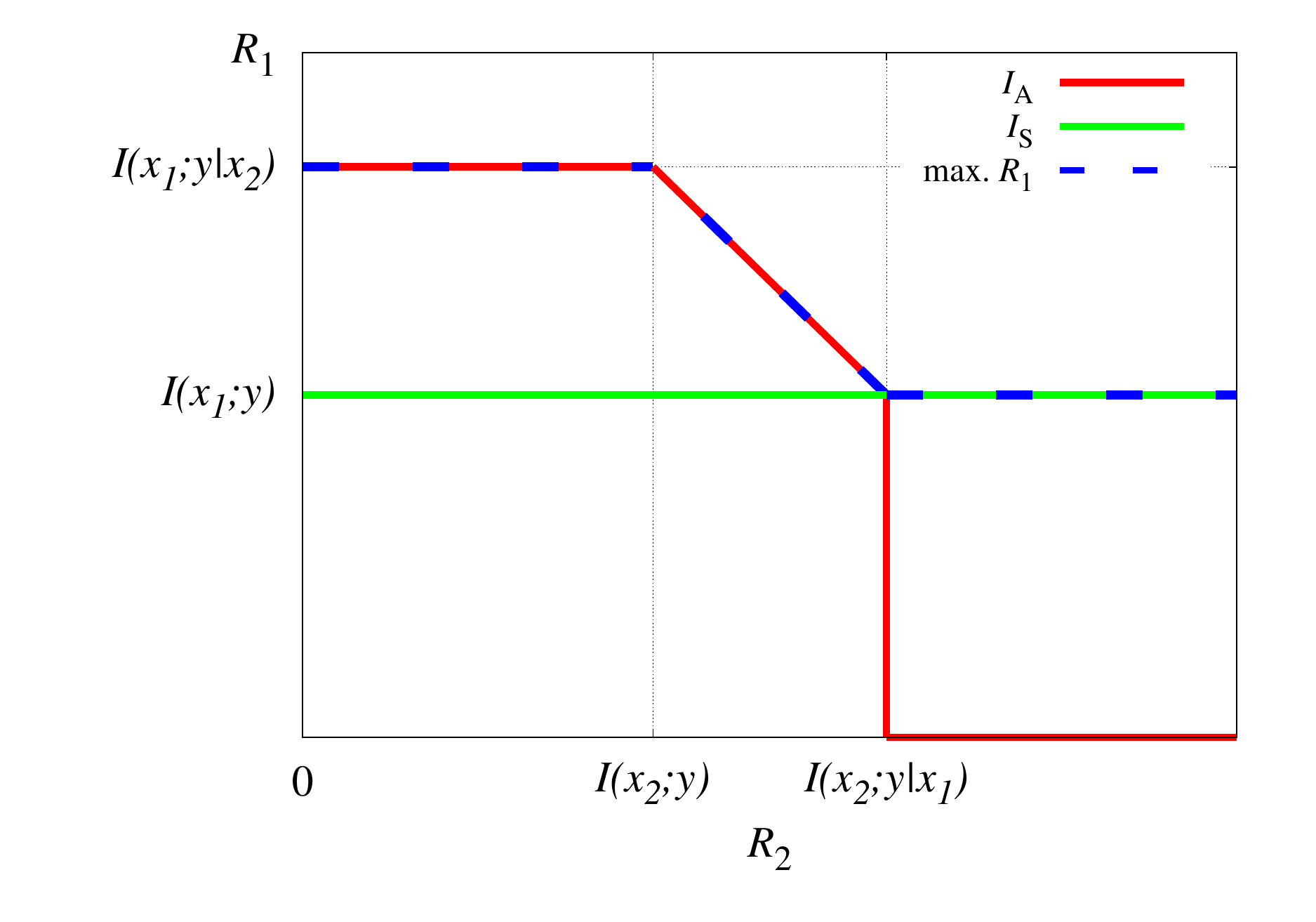}
		\caption{Graphical proof of Theorem~\ref{t:ThIR}.}\label{fig:region_thm}
	\end{center}
\end{figure}

\begin{ex} For Gaussian symbols and $K=2$, we obtain that 
\begin{equation}\nonumber
R_1 < \begin{cases}
		\mathcal{C}\left(\frac{P}{N} \right) &  \mathrm{if\quad}R_2<\mathcal{C}\left(\frac{P\gamma^2}{N+P} \right) \\
	 \mathcal{C}\left(\frac{P(1+\gamma^2)}{N} \right)\!-\!R_2 	& \mathrm{if\quad} \mathcal{C}\left(\frac{P\gamma^2}{N+P} \right) \leq R_2<\mathcal{C}\left(\frac{P\gamma^2}{N} \right)\\
	\mathcal{C}\left(\frac{P}{N+P\gamma^2} \right) & \mathrm{if\quad}R_2\geq \mathcal{C}\left(\frac{P\gamma^2}{N} \right)\,,
      \end{cases}
\end{equation}
where $\mathcal{C}(x)=\log_2(1+x)$.
All curves are shown in Figure \ref{fig:GaussCurve}, for the case of $|\gamma|=0.79$, $R_2=1/2$, and the overall bound is given by the red curve. We can see from the figure that this bound is clearly continuous.
\end{ex}

\begin{figure}
	\includegraphics[width=1.0\columnwidth]{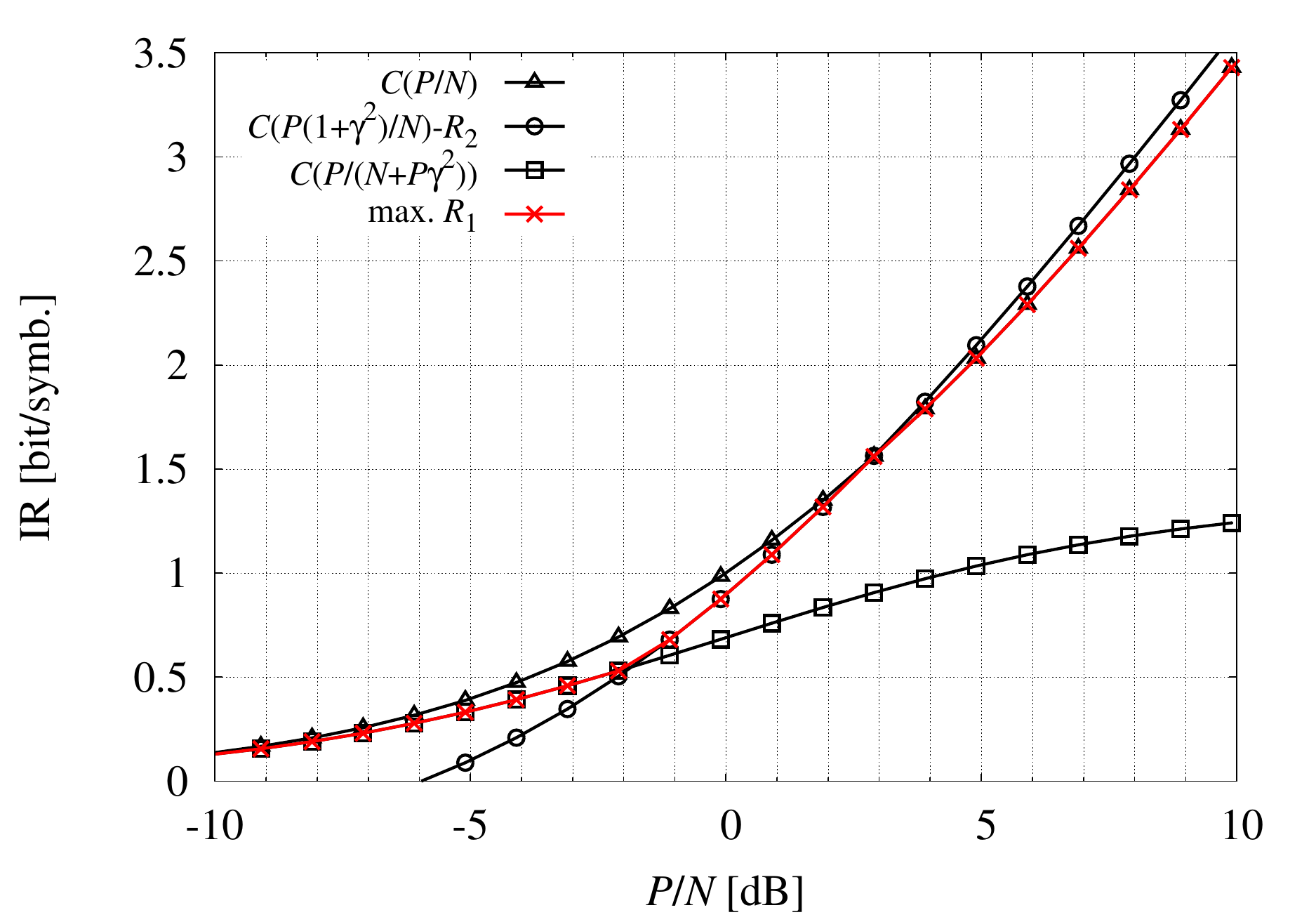}
	\caption{Maximum rate achievable by ``User 1'', for $K=2$, Gaussian symbols, and $R_2=1/2$.}\label{fig:GaussCurve}
\end{figure}

The computation of the IRs $I(x_1;y|x_2)$, $I(x_2;y|x_1)$, $I(x_1,x_2;y)$, $I(x_1;y)$ can be performed by using the achievable lower bound based on 
mismatched detection~\cite{MeKaLaSh94}.

When a SUD is employed at the terminal, the theoretic analysis can be based on the following discrete-time model 
	\begin{equation}\nonumber
		y = x_1+ w \, ,
	\end{equation}
where $w$ includes the thermal noise and the interferers that the receiver ignores.
As known, the complexity of the SUD is much lower than that of the multiuser receiver, and is proportional to $M^{(1)}$.
The computation of the IR $I(x_1;y)$ is again based on mismatched detection~\cite{MeKaLaSh94} and allows us to select the maximum rate for ``User~1'' 
when the co-channel interference is not accounted for.

\section{Information-theoretic Analysis for Scenario~2}\label{s:scenario2}
Let us consider the fraction $\alpha$ of time when both signals are used to send information to ``User 1''. 
Hence, during this time slot both signals $s_1(t)$ and $s_2(t)$ are intended to ``User~1''. 
Since $s_1(t)$ and $s_2(t)$ are independent, we are exactly in the case of the MAC. By properly selecting the rate of the two signals any point of the capacity region can be achieved~\cite{CoTh06}. Clearly, we are interested in selecting the two rates in such a way that the sum-rate $I(x_1,x_2;y)$ is maximized. 

% On the other hand, when signal $s_1(t)$ is a properly phase-shifted version of $s_2(t)$ in order to allow them to sum up coherently, the analysis 
% is very simple since the receiver has to consider only one signal (signal $s_1(t)$ with a proper amplitude) and all interferers $s_i(t)$, $i=3,\dots,K$ 
% are assumed to be additive noise.

\section{Numerical Results}\label{s:num_res}
In this section, we compare the two scenarios described in Section~\ref{s:sys_model} and the corresponding detection strategies by considering 
the performance of ``User~1'', evaluated in terms of IR. 

We assume as reference system the DVB-S2 standard~\cite{Et03} and hence consider adaptive coding and modulation. We choose a frequency reuse with factor two, to generate a high co-channel interference, and consider the users 
which are located close to the edge of the coverage area of the reference beam. In this case, it has been shown \cite{AnAnCaCo14} that it is sufficient to consider the five strongest interfering beams. Therefore, we simulate $K=6$ users, employing different modulation formats. 
Particularly, users with $i=1,\,2$ adopt a QPSK modulation, users with $i=3,\,4$ and 6 adopt a 8PSK modulation, and the user with $i=5$ adopts a 16APSK modulation. 

To identify the users' power profile, we define the signal-to-interference power ratio as
$$
\lambda_i=|\gamma_1|^2/|\gamma_i|^2\, ,
$$
and consider three realistic cases which have a different power profile, and are listed in Table~\ref{t:cases}. These distributions are typical of the forward link of a multibeam broadband satellite system with 2 colors frequency reuse.

\begin{table}  \caption{Power profiles for the considered simulations, corresponding to a two colors frequency reuse.}   \label{t:cases}
% \vspace{-0.3cm}
\begin{center}
\begin{tabular}{|c|c|c|c|c|c|}
  \hline 
  Case & $\lambda_2$ & $\lambda_3$ & $\lambda_4$ & $\lambda_5$ & $\lambda_6$  \\ 
  \hline \hline
   1 &  0 dB & 25 dB & 25 dB & 27 dB & 30 dB \\ 
  \hline 
   2  & 2 dB & 26 dB & 26 dB & 27 dB & 30 dB\\ 
  \hline
   3 & 4 dB & 27 dB & 26 dB & 27 dB & 30 dB \\ 
  \hline 
  \end{tabular}
\end{center}
\end{table}

Figures~\ref{fig:IR1}--\ref{fig:IR3} show the IR related to ``User 1'' as a function of $P/N$ for the three considered interference patterns. In the case of scenario 1, we evaluate both 
the IR achievable by a SUD and that achievable by the MUD$\times$2 algorithm. 
In case of MUD$\times$2, the performance is heavily affected by the rate of ``User~2'', and hence we have to analyze performance for a fixed rate of the binary code employed by the signal $s_2(t)$. We thus have three IR curves for the MUD in scenario~1, to consider the case where ``User~2'' adopts a low code rate (3/5), an average code rate (5/6) and a high rate (8/9), chosen among the ones foreseen by the standard.
In the case of scenario~2, it is assumed that $\alpha=0.5$ and this has been taken into account in the computation of the IR. We recall that for this scenario the relative phase shift of signals $s_1(t)$ and $s_2(t)$ has been optimized by simulation.

Our results show that we cannot identify the strategy which universally achieves the best performance. In particular, the figures show that ``User~1'' has the best IR in  scenario~2 for low-to-medium SNR values in the first case,  where the interference of the second signal is very strong, while in the other cases the advantage of scenario~2 is reduced.
% This is especially true when the signals $s_1(t)$ and $s_2(t)$ are the same, and for the first two cases, where the interference of the second user is very strong,  On the other hand, the only way to achieve very high IRs is to use more involved receivers, while the best scenario depends on the rate of ``User~2''. 

As expected, in scenario~1 the adoption of the MUD gives the best results with respect to the SUD, and this is at the price of an increased complexity. Moreover, the performance of the MUD heavily depends on the rate of the strongest interfering user. In case~3, the SUD gives very good IRs and hence it is the best choice to compromise between complexity and performance for a large SNR range.

\begin{figure}
	\begin{center}
		\includegraphics[width=1.0\columnwidth]{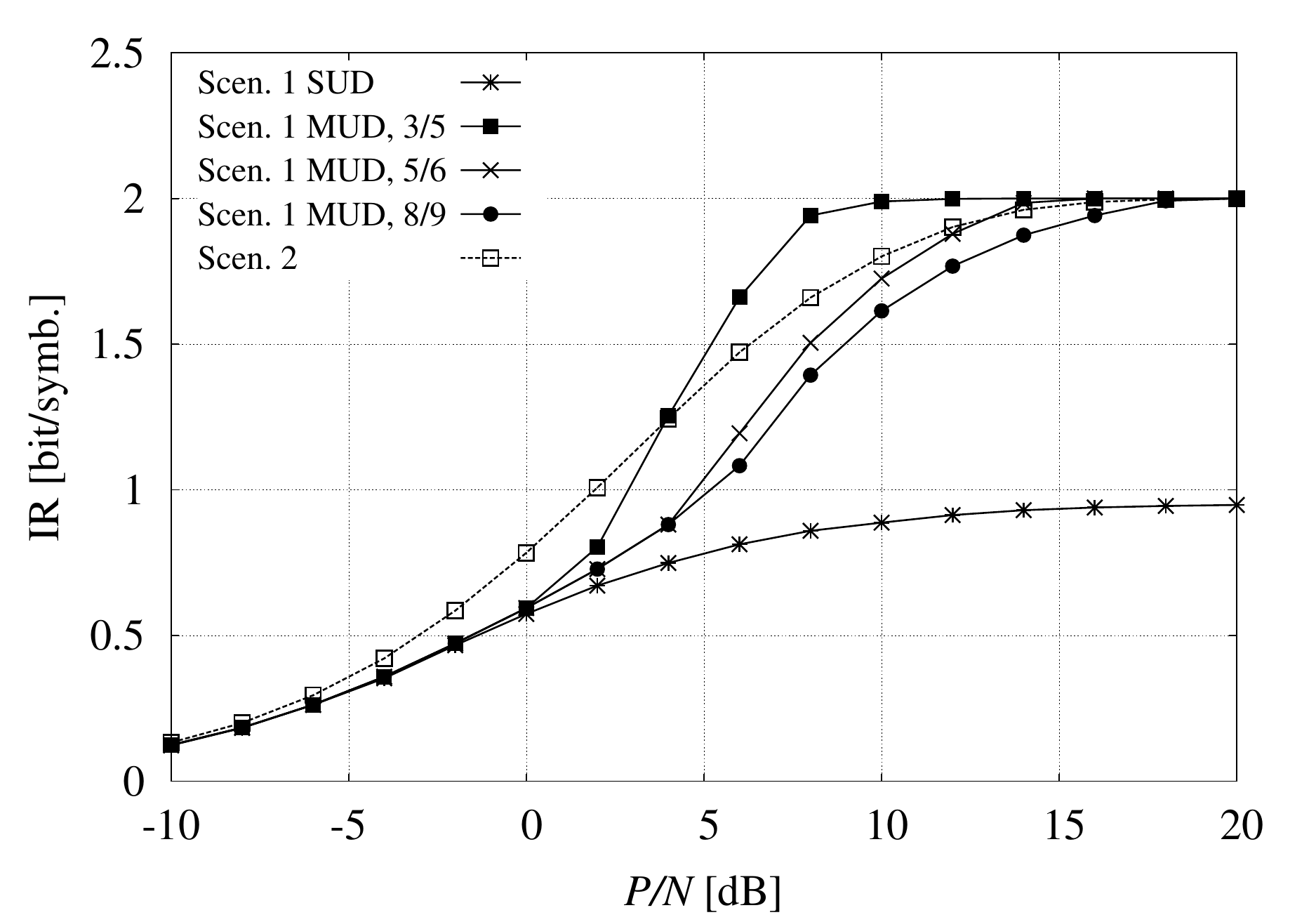}
		\caption{Information rate of ``User 1'' for the case 1 in the considered scenarios, using different receivers.}\label{fig:IR1}
	\end{center}
\end{figure}

\begin{figure}
	\begin{center}
		\includegraphics[width=1.0\columnwidth]{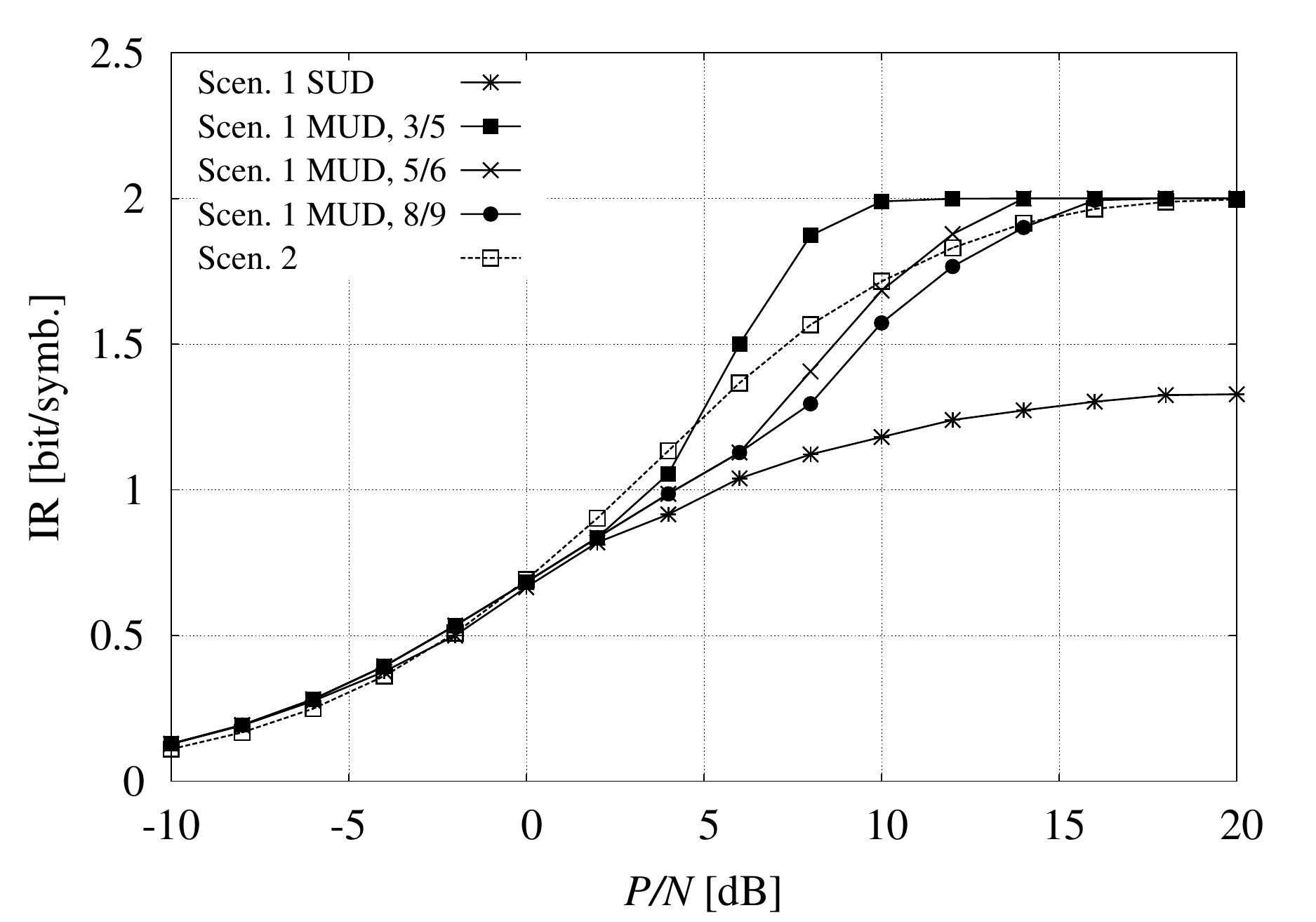}
		\caption{Information rate of ``User 1'' for the case 2 in the considered scenarios, using different receivers.}\label{fig:IR2}
	\end{center}
\end{figure}

\begin{figure}
	\begin{center}
		\includegraphics[width=1.0\columnwidth]{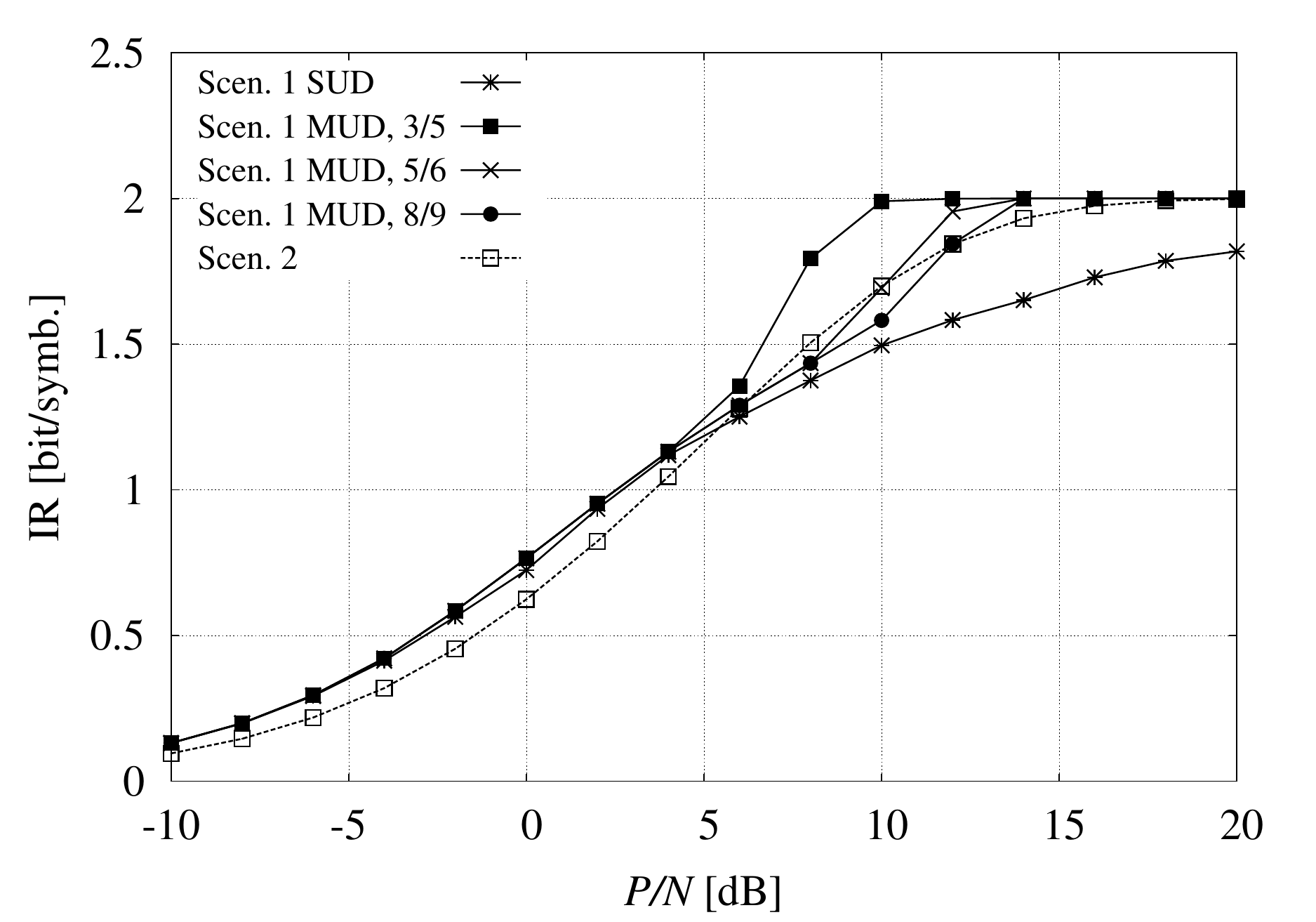}
		\caption{Information rate of ``User 1'' for the case 3 in the considered scenarios, using different receivers.}\label{fig:IR3}
	\end{center}
\end{figure}

\section{Conclusions}\label{s:conclusions}
In this paper we have addressed the problem of multiuser detection in the forward link of a multibeam satellite system, in the presence of strong co-channel interference. We considered alternative techniques to the single user detection, that take into account the strongest interfering signal. We have shown that this technique can considerably increase the achievable rate at the cost of a higher computational complexity. 

Furthermore, we considered a transmission strategy where the signals from two beams serve two users in a time division multiplexing way, and we show that this approach is effective at low signal-to-noise ratio, when the co-channel interference is very strong. However, our results reveals that there is no clear winner. In fact, the best strategy depends on the power profile of the interfering signals, the rates of the signals, and the signal-to-noise power ratio.
\section*{Acknowledgement}
This work is partially funded by the European Space Agency, ESA-ESTEC, Noordwijk, The Netherlands. 
The view expressed herein can in no way be taken to reflect the official opinion of the European Space Agency.

\bibliographystyle{ieeetr}
% \bibliography{Refs}

\end{document}